\documentclass[11pt]{article}
\usepackage{tikz}
\usetikzlibrary{shapes,arrows,shadows,positioning,matrix,patterns,decorations.markings} 
\usepackage{color}
\definecolor{background}{cmyk}{0.2,0,0,0}
\definecolor{textcolor}{cmyk}{1,0,0,0.3}
\usepackage{amssymb}
\usepackage{enumerate}
\usepackage{comment}
\usepackage{a4wide}
\usepackage{color}
\usepackage[mathscr]{eucal}

\newtheorem{theorem}{Theorem}
\newtheorem{lemma}[theorem]{Lemma}
\newtheorem{corollary}[theorem]{Corollary}
\newtheorem{proposition}[theorem]{Proposition}

\newtheorem{definition}[theorem]{Definition}

\newenvironment{proof}{\noindent\textbf{Proof.}}{{}\hfill$\Box$\\}
\newenvironment{proofof}{\noindent\textbf{Proof of~}}{{}\hfill$\Box$\\}

\def\eps{\varepsilon}
\def\cE{\mathscr{E}}
\def\Yao#1#2{\ensuremath{\mathrm{Yao}_{#1}^{#2}}}
\def\Y4{\Yao{4}{\infty}}

\title{\textbf{The Stretch Factor of $L_1$- and $L_\infty$-Delaunay
    Triangulations}}

\author{
  Nicolas Bonichon\\
  LaBRI - INRIA Bordeaux Sud-Ouest\\
  \textrm{bonichon@labri.fr}
  \and
  Cyril Gavoille\thanks{Member of the ``Institut Universitaire
    de France''. Supported also by the ANR project ``ALADDIN''.}\\
  LaBRI - University of Bordeaux\\
  \textrm{gavoille@labri.fr}
  \and
  Nicolas Hanusse\\
  LaBRI - CNRS\\
  \textrm{hanusse@labri.fr} 
  \and
  Ljubomir Perkovi\'{c}\thanks{Supported by a Fulbright Aquitaine
    Regional grant and a DePaul University research grant.}\\
  DePaul University, Chigago\\
  \textrm{lperkovic@cs.depaul.edu}
}

\date{}

\begin{document}

\maketitle

\begin{abstract}
  In this paper we determine the stretch factor of the $L_1$-Delaunay
  and $L_\infty$-Delaunay triangulations, and we show that this
  stretch is $\sqrt{4+2\sqrt{2}} \approx 2.61$. Between any two points
  $x,y$ of such triangulations, we construct a path whose length is no
  more than $\sqrt{4+2\sqrt{2}}$ times the Euclidean distance between
  $x$ and $y$, and this bound is best possible. This definitively
  improves the 25-year old bound of $\sqrt{10}$ by Chew (SoCG~'86).

  To the best of our knowledge, this is the first time the stretch
  factor of the well-studied $L_p$-Delaunay triangulations, for any
  real $p\ge 1$, is determined exactly.
\end{abstract}


\section{Introduction}

Given a set of points $P$ on the plane, the Delaunay triangulation for
$P$ is a spanning subgraph of the Euclidean graph on $P$ that is the
dual of the Vorono{\"\i} diagram of $P$. The Delaunay triangulation is
a fundamental structure with many applications in computational
geometry and other areas of Computer Science. In some applications
(including on-line routing~\cite{BM04b}), the triangulation is used as
a spanner, defined as a spanning subgraph in which the distance
between any pair of points is no more than a constant multiplicative
ratio of the Euclidean distance between the points. The constant ratio
is typically referred to as the stretch factor of the spanner. While
Delaunay triangulations have been studied extensively, obtaining a
tight bound on its stretch factor has been elusive even after decades
of attempts.

In the mid-1980s, it was not known whether Delaunay triangulations
were spanners at all.  In order to gain an understanding of the
spanning properties of Delaunay triangulations, Chew considered
related, ``easier'' structures. In his seminal 1986
paper~\cite{Chew86}, he proved that an $L_1$-Delaunay triangulation
--- the dual of the Vorono{\"\i} diagram of $P$ based on the
$L_1$-metric rather than the $L_2$-metric --- has a stretch factor
bounded by $\sqrt{10}$. Chew then continued on and showed that the a
TD-Delaunay triangulation --- the dual of a Vorono{\"\i} diagram
defined using a \emph{Triangular Distance}, a distance function not
based on a circle ($L_2$-metric) or a square ($L_1$-metric) but an
equilateral triangle --- has a stretch factor of~$2$~\cite{Chew89}.

Finally, Dobkin et al.~\cite{DFS87} succeeded in showing that the
(classical, $L_2$-metric) Delaunay triangulation of $P$ is a spanner
as well. The bound on the stretch factor they obtained was
subsequently improved by Keil and Gutwin~\cite{KG92} as shown in
Table~\ref{ta:related}. The bound by Keil and Gutwin stood unchallenged
for many years until very recently when Xia improved the bound to
below 2~\cite{Xia11}.

\begin{table}
\label{ta:related}
\begin{center}
\mbox{\renewcommand{\arraystretch}{1.35}
\begin{tabular}{llr}
{\bf Paper} &  {\bf Graph} & {\bf Stretch factor} \\ \hline\hline
\cite{DFS87} & $L_2$-Delaunay & $\pi(1+\sqrt{5})/2 \approx 5.08$ \\ \hline
\cite{KG92} & $L_2$-Delaunay & $4\pi/(3\sqrt{3}) \approx 2.41$ \\ \hline
\cite{Xia11} & $L_2$-Delaunay & $1.998$ \\ \hline
\cite{Chew89} & TD-Delaunay & $\mathbf{2}$ \\ \hline
\cite{Chew86} & $L_1$-,$L_\infty$-Delaunay & $\sqrt{10} \approx 3.16$ \\ \hline
\textbf{[this paper]}\hspace{5ex} & $L_1$-,$L_\infty$-Delaunay \hspace{3ex} & $\mathbf{\sqrt{4+2\sqrt{2}}\approx 2.61}$\\ \hline\hline
\end{tabular}}
\end{center}
\caption{Key stretch factor upper bounds (optimal values are bold).}
\end{table}

While progress has been made, none of the techniques developed so far
lead to a tight bound on the spanning ratio. There has been some
progress recently on the lower bound side. The trivial lower bound of
$\pi/2 \approx 1.5707$ has recently been improved to
$1.5846$~\cite{BDLSV11} and then to $1.5932$~\cite{XZ11}.

While much effort has been made on understanding the stretch factor of
Delaunay triangulations, little has been done on the $L_p$-Delaunay
triangulations for $p\neq 2$. Lee and Wong~\cite{LW80} show that
$L_1$-,$L_\infty$-Delaunay triangulations have applications in
scheduling problems for $2$-dimensional storage, and how to construct,
for all real $p\ge 1$, Vorono{\"\i} diagrams in the $L_p$-metric in
$O(n\log{n})$ time~\cite{Lee80}. Delaunay triangulations based on
arbitrary convex distance functions have been studied
in~\cite{BCCS08}, which shows that such geometric graphs are indeed
plane graphs and spanners whose stretch factor depends only on the
shape of the convex body. However, due to the general approach, the
bounds on the stretch factor remain loose. For instance the bounds
they obtain for $L_2$-Delaunay triangulations are $> 24$.

The general picture is that, in spite of much effort, with the exception
of the triangular distance the exact value of the stretch factor of Delaunay
triangulations based on any convex function is unknown. In particular,
the stretch factor of $L_p$-Delaunay triangulations is unknown for each
$p\ge 1$.

\paragraph{Our contributions.}

We show that the exact stretch factor of $L_1$-Delaunay triangulations
and $L_\infty$-Delaunay triangulations is $\sqrt{4+2\sqrt{2}} \approx
2.61$, ultimately improving the $3.16$ bound of Chew~\cite{Chew86}.

Technically, we use rectangular coordinates to prove the upper
bound. We show that the distance in the triangulation between any
source-destination pair of points lying on the border of a horizontal
rectangle of length $x$ and of depth $y\le x$ is no more than
$(1+\sqrt{2})x+y$. The stretch factor bound then simply follows.  In
our proof, we construct the route from the source to the destination
by maintaining two possible short paths, until we reach some special
point (called \emph{inductive} point) where we can apply our main
inductive hypothesis.

Despite the technical aspect of our contribution, we believe that our
proof techniques may give insights into determining the stretch factor of
other convex distance based Delaunay triangulations. For example, let $P_k$
denote the convex distance function defined by a regular $k$-gon. We observe
that the stretch factor of $P_k$-Delaunay triangulations is known for $k=3,4$
since $P_3$ is the triangular distance function of~\cite{Chew89}, and
$P_4$ is nothing else than the $L_\infty$-metric. Determining the
stretch factor of $P_k$-Delaunay triangulations for larger $k$ would
undoubtedly be an important step towards understanding the stretch factor
of classical Delaunay triangulations.


\section{Preliminaries}
\label{sec:prelim}

Given a set $P$ of points in the two-dimensional Euclidean space, the
Euclidean graph $\cE$ is the complete weighted graph embedded in the
plane whose nodes are identified with the points. We assume a
Cartesian coordinate system is associated with the Euclidean space and
thus every point can be specified with $x$ and $y$ coordinates. For
every pair of nodes $u$ and $w$, the edge $(u,w)$ represents the
segment $[uw]$ and its weight is the edge length defined in Euclidean
distance: $d_2(u,w) = \sqrt{d_x(u,w)^2 + d_y(u,w)^2}$ where $d_x(u,w)$
(resp. $d_y(u,w)$) is the difference between the $x$ (resp. $y$)
coordinates of $u$ and $w$. 

We say that a subgraph $H$ of a graph $G$ is a \emph{$t$-spanner} of
$G$ if for any pair of vertices $u,v$ of $G$, the distance between $u$
and $v$ in $H$ is at most $t$ times the distance between $u$ and $v$
in $G$; the constant $t$ is referred to as the \emph{stretch factor}
of $H$ (with respect to $G$).  $H$ is a $t$-spanner (or spanner for
some $t$ constant) if it is a $t$-spanner of $\cE$.

In our paper, we deal with the construction of spanners based on
Delaunay triangulations. As we saw in the introduction, the
$L_1$-Delaunay triangulation is the dual of the Vorono{\"\i} diagram
based on the $L_1$-metric $d_1(u,w) = d_x(u,w) + d_y(u,w)$. A property
of the $L_1$-Delaunay triangulations, actually shared by all
$L_p$-Delaunay triangulations, is that all their triangles can be
defined in terms of empty circumscribed convex bodies (squares for
$L_1$ or $L_\infty$ and circles for $L_2$). More precisely, let a
\emph{square} in the plane be a square whose sides are parallel to the
$x$ and $y$ axis and let a \emph{tipped square} be a square tipped at
$45^\circ$. For every pair of points $u,v \in P$, $(u,v)$ is an edge
in the \emph{$L_1$-Delaunay triangulation} of $P$ iff there is a
tipped square that has $u$ and $v$ on its boundary and contains no
point of $P$ in its interior (cf.~\cite{Chew89}).


If a \emph{square} with sides parallel to the $x$ and $y$ axes, rather
than a tipped square, is used in this definition then a different
triangulation is defined; it corresponds to the dual of the
Vorono{\"\i} diagram based on the $L_\infty$-metric $d_\infty(u,w) =
\max\{d_x(u,w), d_y(u,w)\}$. We refer to this triangulation as the
$L_\infty$-Delaunay triangulation. This triangulation is nothing more
than the $L_1$-Delaunay triangulation of the set of points $P$ after
rotating all the points by $45^\circ$ around the origin. Therefore
Chew's bound of $\sqrt{10}$ on the stretch factor of the
$L_1$-Delaunay triangulation (\cite{Chew86}) applies to
$L_\infty$-Delaunay triangulations as well. In the remainder of this
paper, we will be referring to $L_\infty$-Delaunay (rather than $L_1$)
triangulations because we will be (mostly) using the $L_\infty$-metric
and squares, rather than tipped squares.

One issue with Delaunay triangulations is that there might not be a
unique triangulation of a given set of points $P$. To insure
uniqueness and keep our arguments simple, we make the usual assumption
that the points in $P$ are in \emph{general position}, which for us
means that no four points lie on the boundary of a square and
no two points share the same abscissa or the same ordinate.

\begin{figure}[b!]
\begin{center}
\begin{tikzpicture}

\draw [color=textcolor] (0,0) rectangle (2.15, 1.22);

\node at (0, 0) [fill, inner sep = 2pt, circle, label=180:$a$] (l) {};
\node at (2.01, -1.8) [fill, inner sep = 2pt, circle, label=0:$c_2$] {};

\foreach \x in {0.2, 0.4, ..., 3.0}
{
  \node at (\x*0.05+2, \x-2) [fill, inner sep = 1pt, circle] (h) {};
  \draw (h) -- (l);
  \draw (h) -- (\x*0.05+2.01, \x-1.8);
  \node at (\x*0.05, \x) [fill, inner sep = 1pt, circle] (l) {};
  \draw (h) -- (l);
  \draw (l) -- (\x*0.05-0.01, \x-0.2);
}

\node at (0.15, 3) [fill, inner sep = 2pt, circle, label=180:$c_1$] {};
\node at (2.15, 1.2) [fill, inner sep = 2pt, circle, label=0:$b$] (h) {};
\draw (h) -- (l);

\draw [color=textcolor, dashed] (0.15, 3) -- (0.15, 1) -- (2.15, 1) -- node [right] {\scriptsize $1-2\delta$} (2.15, 3) -- node [above] {\scriptsize $1-\delta$} (0.15, 3);

\node at (-0.8, 1.5) [] {\color{textcolor}{\scriptsize $\sqrt{2}-2\delta$}};
\node at (2.75, 0.6) [] {\color{textcolor}{\scriptsize $\sqrt{2}-1$}};
\node at (2.65, -0.9) [] {\color{textcolor}{\scriptsize $1-2\delta$}};
\end{tikzpicture}\hspace{2cm}
\begin{tikzpicture}
\draw [color=textcolor, dashed] (0,0) rectangle (3, -3);
\node at (3.2, 0.2) [] {\color{textcolor}{$S_1$}};

\node at (3,-3) [fill, circle, inner sep=2pt, label=0:$c_2$] (c2) {};
\node at (3.25,-2.3) [fill, circle, inner sep=2pt, label=0:$q_1$] (q) {};
\node at (0,-0.7) [fill, circle, inner sep=2pt, label=180:$a$] (a) {};
\node at (0.25,0) [fill, circle, inner sep=2pt, label=115:$p_1$] (p1) {};
\node at (0.5,0.7) [fill, circle, inner sep=2pt, label=115:$p_2$] (p2) {};

\draw (q) -- (p2) -- (p1) -- (c2) -- (q) -- (p1) -- (a) -- (c2);

\node at (-0.24,-0.24) [color=textcolor] {{\small $\delta$}};
\end{tikzpicture}

a) \hspace{6cm} b)
\end{center}
\caption{a) An $L_\infty$-Delaunay triangulation with stretch factor arbitrarily
close to $\sqrt{4+2\sqrt{2}}$. b) A closer look at the first few faces of
this triangulation.}
\label{fi:L1_pire}
\end{figure}
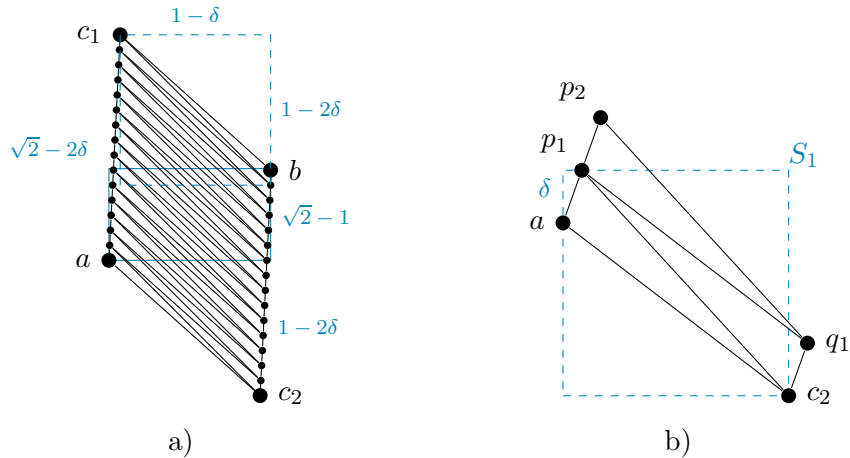

We end this section by giving a lower bound on the stretch factor of
$L_\infty$-Delaunay triangulations.

\begin{proposition}
  \label{le:lower_bound}
  For every $\eps > 0 $, there exists a set of points $P$ in the
  plane such that the $L_\infty$-Delaunay triangulation on $P$ has
  stretch factor at least
$$\sqrt{4+2\sqrt{2}}-\eps~.$$
\end{proposition}

This lower bound applies, of course, to $L_1$-Delaunay triangulations
as well.  The proof of this proposition relies on the example of
Figure~\ref{fi:L1_pire}.

\begin{proof}
  Given $\delta > 0$, we define the set of points $P$ as follows.  Let
  point $a$ be the origin and let points $b$, $c_1$, and $c_2$ have
  coordinates $(1, \sqrt{2}-1)$, $(\delta, \sqrt{2}-2\delta)$, and
  $(1-\delta, 1-2\delta)$, respectively. Additional $k =
  {{\sqrt{2}-2\delta} \over \delta} - 1$ points are placed on line
  segment $[ac_1]$ and another $k$ on line segment $[c_2b]$ in such a
  way that the difference in $y$ coordinates between successive points
  on a segment is $\delta$, as shown in
  Figures~\ref{fi:L1_pire}. (W.l.o.g. assume that $\sqrt{2} \over
  \delta$ and thus $k$ is an integer so that this can be done.)  Let
  $a = p_0, p_1, p_2, p_3, \dots, p_k, p_{k+1}=c_1$ be the labels, in
  order as they appear when moving from $a$ to $c_1$, of the points on
  segment $[ac_1]$ and let $c_2=q_0, q_1, q_2, q_3, \dots, q_{k+1}=b$
  be the labels, in order as they appear when moving from $c_2$ to
  $b$, of the points on segment $[c_2b]$, as illustrated in
  Figure~\ref{fi:L1_pire}.

  Consider the square $S_1$ of side length $1-\delta$ and having $a$
  and $p_1$ on its west (left) and north sides, respectively (see
  Figure~\ref{fi:L1_pire}b)). Since $d_\infty(a,c_2) = d_x(a,c_2) =
  1-\delta$ and $d_\infty(p_1,c_2) = d_y(p_1,c_2) = 1-\delta$, point
  $c_2$ is exactly the southeast vertex of square $S_1$.  By symmetry,
  it follows that for every $i = 0, 1, 2, \dots, k$, if $S_i$ is the
  square of side length $1-\delta$ with $p_i$ and $p_{i+1}$ on its
  west and north sides, then point $q_i$ is exactly the southeast
  vertex of $S_i$.  This means that all points $q_j$ with $j \not= i$
  as well as all points $p_j$ with $j \not= i, i+1$ must lie outside
  $S_i$. Therefore, for every $i = 0, 1, 2, \dots, k$, points $p_i$,
  $p_{i+1}$, and $q_i$ define a triangle in the $L_\infty$-Delaunay
  triangulation $T$ on $P$. A similar argument shows that the path
  $q_0, q_1, \dots, q_{k+1}$ is in triangulation $T$ as well.  The
  triangulation $T$ is illustrated in Figure~\ref{fi:L1_pire}a).

  Having defined the set of points $P$ and described its
  $L_\infty$-Delaunay triangulation $T$, we now analyze the stretch
  factor of $T$. A shortest path from $a$ to $b$ in $T$ is, for
  example, $a, p_1, p_2, \dots, p_k, c_1, b$. The length of this path
  is
  \begin{eqnarray*}
    d_2(a,c_1) + d_2(c_1,b) & = & \sqrt{d_x(a,c_1)^2+d_y(a,c_1)^2} + \sqrt{d_x(c_1,b)^2+d_y(c_1,b)^2} \\
    & = & \sqrt{(\sqrt{2}-\delta)^2 + \delta^2} + \sqrt{(1-\delta)^2 + (1-2\delta)^2}
  \end{eqnarray*}
  which tends to $2\sqrt{2}$ as $\delta$ tends to $0$. The Euclidean
  distance between $a$ and $b$ is:
$$
d_2(a,b) = \sqrt{1^2+(\sqrt{2}-1)^2} = \sqrt{4-2\sqrt{2}}~.
$$
Therefore, it is possible to construct a $L_\infty$-Delaunay
triangulation whose stretch factor is arbirtrarily close to:
$$
{2\sqrt{2} \over \sqrt{4-2\sqrt{2}}} = \sqrt{4+2\sqrt{2}}~.
$$
\end{proof}


\section{Main result}

In this section we obtain a tight upper bound on the stretch factor of
an $L_\infty$-Delaunay triangulation. It follows from this
key theorem: 

\begin{theorem}
\label{th:squareDel}
Let $T$ be the $L_\infty$-Delaunay triangulation on a set of points
$P$ in the plane and let $a$ and $b$ be any two points of $P$. If $x =
d_\infty(a, b) = \max \{d_x(a,b), d_y(a,b)\}$ and $y = \min
\{d_x(a,b), d_y(a,b)\}$ then
$$d_T(a, b) ~\le~ (1+\sqrt{2})x + y$$
where $d_T(a,b)$ denotes the distance between $a$ and $b$ in
triangulation $T$.
\end{theorem}

\begin{corollary}
  The stretch factor of the $L_1$- and the 
  $L_\infty$-Delaunay triangulation on a set of points $P$ is at most
  $$\sqrt{4+2\sqrt{2}} ~\approx~ 2.6131259\dots$$
\end{corollary}

\begin{proof}
  By Theorem~\ref{th:squareDel}, an upper-bound of the stretch factor
  of an $L_\infty$-Delaunay triangulation is the maximum of the function
  $$
  {{(1+\sqrt{2})x + y} \over \sqrt{x^2+y^2}}
  $$
  over values $x$ and $y$ such that $0 < y \leq x$. The maximum is
  reached when $x$ and $y$ satisfy $y/x = 1+\sqrt{2}$, and the maximum
  is equal to $\sqrt{1 + (1+\sqrt{2})^2} = \sqrt{4+2\sqrt{2}}$. As
  $L_1$- and $L_\infty$-Delaunay triangulations have the same stretch
  factor, this result also holds for $L_1$-Delaunay triangulations.
\end{proof}

To prove Theorem~\ref{th:squareDel}, we will construct a bounded
length path in $T$ between two arbitrary points $a$ and $b$ of $P$.
To simplify the notation and the discussion, we assume that point $a$
has coordinates $(0, 0)$ and point $b$ has coordinates $(x, y)$ with
$0 < y \leq x$. The segment $[ab]$ divides the Euclidean plane into
two half-planes; a point in the same half-plane as point $(0, 1)$ is
said to be \emph{above} segment $[ab]$, otherwise it is
\emph{below}. Let $T_1, T_2, T_3, \dots, T_k$ be the sequence of
triangles of triangulation $T$ that line segment $[ab]$ intersects
when moving from $a$ to $b$.  Let $h_1$ and $l_1$ be the nodes of
$T_1$ other than $a$, with $h_1$ lying above segment $[ab]$ and $l_1$
and lying below. Every triangle $T_i$, for $1 < i < k$, intersects
line segment $[ab]$ twice; let $h_i$ and $l_i$ be the endpoints of the
edge of $T_i$ that intersects segment $[ab]$ last, when moving on
segment $[ab]$ from $a$ to $b$, with $h_i$ being above and $l_i$ being
below segment $[ab]$. Note that either $h_i = h_{i-1}$ and $T_i =
\triangle(h_i, l_i, l_{i-1})$ or $l_i = l_{i-1}$ and $T_i =
\triangle(h_{i-1}, h_i, l_i)$, for $1 < i < k$. We also set $h_0 = l_0
= a$, $h_k = b$, and $l_k=l_{k-1}$. For $1 \leq i \leq k$, we define
$S_i$ to be the square whose sides pass through the three vertices of
$T_i$ (see Figure~\ref{fi:squares}); since $T$ is an
$L_\infty$-Delaunay triangulation, the interior of $S_i$ is devoid of
points of $P$.  We will refer to the sides of the square using the
notation: N (north), E (east), S (south), and W (west). We will also
use this notation to describe the position of an edge connecting two
points lying on two sides a square: for example, a WN edge connects a
point on the west and a point on the N side.  We will say that an edge
is \emph{gentle} if the line segment corresponding to it in the graph
embedding has a slope within $[-1, 1]$; otherwise we will say that it
is \emph{steep}.

\begin{figure}[b!]
\begin{center}
\begin{tikzpicture}

\draw [color=textcolor, dashed] (0,0) rectangle (1.5, 1.5);
\draw (0.2, 1.3) node [label=135:{\color{textcolor}{$S_1$}}] {};

\draw [color=textcolor, dashed] (0.5,0) rectangle (2.75, 2.25);
\draw (0.7, 2.05) node [label=135:{\color{textcolor}{$S_2$}}] {};

\draw [color=textcolor, dashed] (1, -.75) rectangle (4, 2.25);
\draw (3.9, 2.15) node [label=45:{\color{textcolor}{$S_3$}}] {};

\draw [color=textcolor] (0,0.35) rectangle (6, 1.15);

\draw (6, 1.15) node [fill, circle, inner sep=2pt, label=0:$b$] (b) {};

\draw (0,0.35) node [fill, circle, inner sep=2pt, label=180:$a$] (a) {};

\draw (0.5,1.5) node [fill, circle, outer sep=-4pt, inner sep=2pt, label=135:$h_1$] (h1) {};
\draw (1,0) node [fill, circle, inner sep=2pt, label=225:$l_{1,2}$] (l1) {};
\draw (2, 2.25) node [fill, circle, inner sep=2pt, label=90:$h_{2,3}$] (h2) {};

\draw (4,-.5) node [fill, circle, inner sep=2pt, label=0:$l_3$] (l2) {};
\draw (a) -- (h1) -- (h2) -- (l2) -- (l1) -- (a);
\draw (h1) -- (l1) -- (h2) -- (l2);

\end{tikzpicture}
\end{center}
\caption{Triangles $T_1$ (with points $a$, $h_1$, $l_1$), $T_2$ (with points
$h_1$, $h_2$, and $l_2$), and $T_3$ (with points $l_2$, $h_3$, and $l_3$) and
associated squares $S_1$, $S_2$, and $S_3$. When traveling from $a$ to $b$
along segment $[a,b]$, the edge that is hit when leaving $T_i$ is
$(h_i, l_i)$.}
\label{fi:squares}
\end{figure}
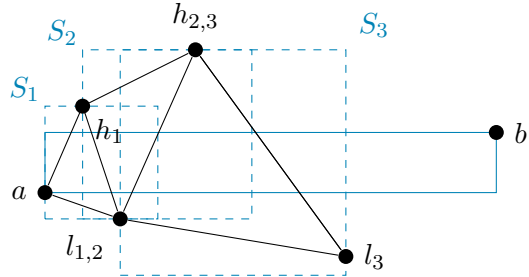

We will prove Theorem~\ref{th:squareDel} by induction on the distance,
using the $L_\infty$-metric, between $a$ and $b$. Let $R(a,b)$ be the
rectangle with sides parallel to the $x$ and
$y$ axes and with vertices at points $a$ and $b$. If there is a point of
$P$ inside $R(a,b)$, we will easily apply induction. The case when
$R(a,b)$ does not contain points of $P$ --- and in particular the
points $h_i$ and $l_i$ for $0 < i < k$ --- is more difficult and we
need to develop tools to handle it. The following Lemma describes the
structure of the triangles $T_1, \dots, T_k$ when $R(a,b)$ is
empty. We need some additional terminology first though: we say that a
point $u$ is \emph{above} (resp. \emph{below}) $R(a,b)$ if $0 < x_u < x$
and $y_u > y$ (resp. $y_u < 0$).

\begin{lemma}
  \label{le:properties}
  If $(a,b) \not\in T$ and no point of $P$ lies inside rectangle
  $R(a,b)$, then point $a$ lies on the W side of square $S_1$, point
  $b$ lies on the E side of square $S_k$, points $h_1, \dots, h_k$ all
  lie above $R(a, b)$, and points $l_1, \dots, l_k$ all lie below
  $R(a, b)$. Furthermore, for any $i$ such that $1 < i < k$:
\begin{enumerate}[\indent a)]

\item Either $T_i = \triangle(h_{i-1}, h_i, l_{i-1}=l_i)$, points $h_{i-1}$,
$h_i$, and $l_{i-1}=l_i$ lie on the sides of $S_i$ in clockwise order,
and $(h_{i-1}, h_i)$ is a WN, WE, or NE edge in $S_i$

\item Or $T_i = \triangle(h_{i-1}=h_i, l_{i-1}, l_i)$, points $h_{i-1}=h_i$,
$l_i$, and $l_{i-1}$ lie on the sides of $S_i$ in clockwise order, and
$(l_{i-1}, l_i)$ is a WS, WE, or SE edge in $S_i$.

\end{enumerate}
\end{lemma}
These properties are illustrated in Figure~\ref{fi:squares}. 

\begin{proof}
  Since points of $P$ are in general position, points $a$, $h_1$, and
  $l_1$ must lie on 3 different sides of $S_1$. Because segment $[ab]$
  intersects the interior of $S_1$ and since $a$ is the origin and $b$
  is in cone 0 of $a$, $a$ can only lie on the W or S side of
  $S_1$. If $a$ lies on the S side then $l_1 \not= b$ would have to
  lie inside $R(a, b)$, which is a contradiction. Therefore $a$ lies
  on the W side of $S_1$ and, similarly, $b$ lies on the E side of
  $S_k$.

  Since points $h_i$ ($0 < i < k$) are above segment $[ab]$ and points
  $l_i$ ($0 < i < k$) are below segment $[ab]$, and because all
  squares $S_i$ ($0 < i < k$) intersect $[ab]$, points $h_1, \dots,
  h_k$ all lie above $R(a, b)$, and points $l_1, \dots, l_k$ all lie
  below $R(a, b)$.

  The three vertices of $T_i$ can be either $h_i=h_{i-1}$, $l_{i-1}$,
  and $l_i$ or $h_{i-1}$, $h_i$, and $l_{i-1} = l_i$. Because points
  of $T$ are in general position, the three vertices of $T_i$ must
  appear on three different sides of $S_i$.  Finally, because
  $h_{i-1}$ and $h_i$ are above $R(a,b)$, they cannot lie on the S
  side of $S_i$, and because $l_{i-1}$ and $l_i$ are below $R(a,b)$,
  they cannot lie on the N side of $S_i$.

  If $T_i = \triangle(h_{i-1}, h_i, l_{i-1}=l_i)$, points $h_{i-1}$,
  $h_i$, $l_i$ must lie on the sides of $S_i$ in clockwise order. The
  only placements of points $h_{i-1}$ and $h_i$ on the sides of $S_i$
  that satisfy all these constraints are as described in {\em a)}.  If
  $T_i = \triangle(h_{i-1}=h_i, l_{i-1}, l_i)$, points $h_i$, $l_i$,
  $l_{i-1}$ must lie on the sides of $S_i$ in clockwise order. Part
  {\em b)} lists the placements of points $l_{i-1}$ and $l_i$ that
  satisfy the constraints.
\end{proof}

In the following definition, we define the points on which
induction can be applied in the proof of Theorem~\ref{th:squareDel}.
\begin{definition}
Let $R(a,b)$ be empty. Square $S_j$ is \emph{inductive} if edge
$(l_j, h_j)$ is gentle. The point $c = h_j$ or $c = l_j$ with the larger
abcissa is the \emph{inductive} point of inductive square $S_j$.
\end{definition}

The following lemma will be the key ingredient of our inductive proof
of Theorem~\ref{th:squareDel}.

\begin{lemma}
  \label{le:inductive}
  Assume that $R(a,b)$ is empty.  If no square $S_1, \dots, S_k$ is
  inductive then $$d_T(a, b) ~\le~ (1 + \sqrt{2}) x + y~.$$ Otherwise
  let $S_j$ be the first inductive square in the sequence $S_1, S_2,
  \dots, S_k$. If $h_j$ is the inductive point of $S_j$ then
$$d_T(a, h_j) + (y_{h_j}-y) ~\le~ (1 + \sqrt{2}) x_{h_j}~.$$
If $l_j$ is the inductive point of $S_j$ then
$$d_T(a, l_j) - y_{l_j} ~\le~ (1 + \sqrt{2}) x_{l_j}~.$$
\end{lemma}

Given an inductive point $c$, we can use use Lemma~\ref{le:inductive}
to bound $d_T(a,b)$ and then apply induction to bound $d_T(b,c)$,
\emph{but only if} the position of point $c$ relative to the position
of point $b$ is \emph{good}, i.e., if $x-x_c \geq |y-y_c|$. If that is
not the case, we will use the following Lemma:
\begin{lemma}
\label{le:mono}
Let $R(a,b)$ be empty and let the coordinates of point $c = h_i$ or $c = l_i$
satisfy $0 < x-x_c < |y-y_c|$.
\begin{enumerate}[\indent a)]

\item If $c = h_i$, and thus $0 < x-x_{h_i} < y_{h_i}-y$, then there exists $j$,
with $i < j \leq k$ such that all edges in path $h_i, h_{i+1}, h_{i+2}, \dots,
h_j$ are NE edges in their respective squares and $x-x_{h_j} \geq y_{h_j}-y
\geq 0$.

\item If $c = l_i$, and thus $0 < x-x_{l_i} < y-y_{l_i}$, then there exists
$j$, with $i < j \leq k$ such that all edges in path $l_i, l_{i+1}, l_{i+2},
\dots, l_j$ are SE edges and $x-x_{l_j} \geq y-y_{l_j} \geq 0$.
\end{enumerate}
\end{lemma}

\begin{proof}
We only prove the case $c = h_j$ as the case $c = l_i$ follows using a
symmetric argument.

We construct the path $h_i, h_{i+1}, h_{i+2}, \dots, h_j$ iteratively.
If $h_i = h_{i+1}$, we just continue building the path from $h_{i+1}$.
Otherwise, $(h_i, h_{i+1})$ is an edge of $T_{i+1}$ which, by
Lemma~\ref{le:properties}, must be a WN, WE, or NE edge in square $S_{i+1}$.
Since the S side of square $S_{i+1}$ is below $R(a, b)$ and because 
$x-x_{h_i} < y_{h_i}-y$, point $h_i$ cannot be on the W side of $S_{i+1}$
(otherwise $b$ would be inside square $S_{i+1}$). Thus $(h_i, h_{i+1})$ is a
NE edge. If $x-x_{h_{i+1}} \geq y_{h_{i+1}}-y$ we stop, otherwise we continue
the path construction from $h_{i+1}$.
\end{proof}

We can now prove the main theorem.

\begin{proofof}{\bf Theorem~\ref{th:squareDel}.}
  The proof is by induction on the distance, using the
  $L_\infty$-metric, between points of $P$ (since $P$ is finite there
  is only a finite number of distances to consider).

  Let $a$ and $b$ be the two points of $P$ that are the closest
  points, using the $L_\infty$-metric. We assume w.l.o.g. that $a$ has
  coordinates $(0, 0)$ and $b$ has coordinates $(x, y)$ with $0 < y
  \leq x$. Since $a$ and $b$ are the closest points using the
  $L_\infty$-metric, the largest square having $a$ as a southwest
  vertex and containing no points of $P$ in its interior, which we
  call $S_a$ must have $b$ on its E side. Therefore $(a, b)$ is an
  edge in $T$ and $d_T(a, b) = d_2(a, b) \leq x + y \leq
  (1+\sqrt{2})x+y$.

  For the induction step, we again assume, w.l.o.g., that $a$ has
  coordinates $(0, 0)$ and $b$ has coordinates $(x, y)$ with $0 < y
  \leq x$.

\noindent {\bf Case 1: $R(a,b)$ is not empty.}

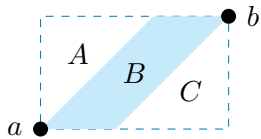
\begin{figure}[b!]
\begin{center}
\begin{tikzpicture}
\draw [color=textcolor,dashed] (0,0) rectangle (2.5, 1.5);

\fill [color=background] (0,0) -- (1.5, 1.5) -- (2.5,1.5) -- (1,0) -- (0,0);

\draw (0.5,1) node (a) {$A$};
\draw (1.25,0.75) node (a) {$B$};
\draw (2,0.5) node (a) {$C$};

\draw (0,0) node [fill, circle, inner sep=2pt, label=180:$a$] (a) {};
\draw (2.5,1.5) node [fill, circle, inner sep=2pt, label=0:$b$] (b) {};
\end{tikzpicture}
\end{center}
\caption{Partition of $R(a,b)$ into three regions in Case 1 of the proof
  of Theorem~\ref{th:squareDel}.}
\label{fi:cinside}
\end{figure}
We first consider the case when there is at least one point of $P$ lying inside
rectangle $R(a, b)$. If there is a point $c$ inside $R(a,b)$ such that
$y_c \leq x_c$ and $y-y_c \leq x-x_c$ (i.e., $c$ lies in the region $B$
shown in Figure~\ref{fi:cinside} then we can apply induction to get
$d_T(a, c) \leq (1 + \sqrt{2}) x_c + y_c$ and $d_T(c,b) \leq (1 + \sqrt{2})
(x- x_c) + y-y_c$ and use these to obtain the desired bound for $d_T(a, b)$.

We now assume that there is no point inside region $B$. If there is
still a point in $R(a,b)$ then there must be one that is on the border
of $S_a$, the square we defined in the basis step, or $S_b$, defined
as the largest square having $b$ as a northeast vertex and containing
no points of $P$ in its interior.  W.l.o.g., we assume the former and
thus there is an edge $(a,c) \in T$ such that either $y_c > x_c$
(i.e., $c$ is inside region $A$ shown in Figure~\ref{fi:cinside} or
$y-y_c > x-x_c$ (i.e., $c$ is inside region $C$).  Either way, $d_T(a,
c) = d_2(a, c) \leq x_c + y_c$. If $c$ is in region $A$, since $x-x_c
\geq y-y_c$, by induction we also have that $d_T(c, b) \leq (1 +
\sqrt{2}) (x-x_c) + (y-y_c)$. Then
\begin{eqnarray*}
d_T(a, b) & \leq & d_T(a,c) + d_T(c, b) \\
          & \leq & x_c + y_c + (1 + \sqrt{2}) (x-x_c) + (y-y_c) \leq  (1 + \sqrt{2}) x + y
\end{eqnarray*}
In the second case, since $x-x_c < y-y_c$, by induction we have that
$d_T(c, b) \leq (1 + \sqrt{2}) (y-y_c) + (x-x_c)$. Then, because $y < x$,
\begin{eqnarray*}
d_T(a, b) & \leq & d_T(a,c) + d_T(c, b) \\
          & \leq & x_c + y_c + (1 + \sqrt{2}) (y-y_c) + (x-x_c) \leq (1 + \sqrt{2}) x + y
\end{eqnarray*}

\noindent {\bf Case 2: The interior of $R(a,b)$ is empty.}

If no square $S_1, S_2, \dots, S_k$ is inductive, $d_T(a,b) \leq (1 +
\sqrt{2}) x + y$ by Lemma~\ref{le:inductive}.  Otherwise, let $S_i$ be
the first inductive square in the sequence and suppose that $h_i$ is
the inductive point of $S_i$. By Lemma~\ref{le:mono}, there is a $j$,
$i \leq j \leq k$, such that $h_i, h_{i+1}, h_{i+2}, \dots, h_j$ is a
path in $T$ of length at most $(x_{h_j}-x_{h_i}) + (y_{h_i}-y_{h_j})$
and such that $x-x_{h_j} \geq y_{h_j}-y \geq 0$. Since $h_j$ is closer
to $b$, using the $L_\infty$-metric, than $a$ is, we can apply
induction to bound $d_T(h_j, b)$. Putting all this together with
Lemma~\ref{le:inductive}, we get:
\begin{eqnarray*}
d_T(a, b) & \leq & d_T(a, h_i) + d_T(h_i, h_j) + d_T(h_j, b) \\
          & \leq & (1 + \sqrt{2})x_{h_i} - (y_{h_i} - y) + (x_{h_j}-x_{h_i}) + (y_{h_i}-y_{h_j}) + (1 + \sqrt{2})(x-x_{h_j}) + (y_{h_j} - y) \\
          & \leq & (1 + \sqrt{2}) x~.
\end{eqnarray*}

If $l_i$ is the inductive point of $S_i$, by Lemma~\ref{le:mono} there
is a $j$, $i \leq j \leq k$, such that $l_i, l_{i+1}, l_{i+2}, \dots,
l_j$ is a path in $T$ of length at most $(x_{h_j}-x_{h_i}) +
(y_{h_j}-y_{h_i})$ and such that $x-x_{h_j} \geq y-y_{h_j} \geq
0$. Because the position of $j$ with respect to $b$ is good and since
$l_j$ is closer to $b$, using the $L_\infty$-metric, than $a$ is, we
can apply induction to bound $d_T(l_j, b)$. Putting all this together
with Lemma~\ref{le:inductive}, we get:
\begin{eqnarray*}
d_T(a, b) & \leq & d_T(a, l_i) + d_T(l_i, l_j) + d_T(l_j, b) \\
          & \leq & (1 + \sqrt{2})x_{l_i} + y_{l_i} + (x_{l_j}-x_{l_i}) + (y_{l_j}-y_{l_i}) + (1 + \sqrt{2})(x-x_{l_j}) + (y - y_{l_j}) \\
          & \leq & (1 + \sqrt{2}) x + y~.
\end{eqnarray*}

\end{proofof}

What remains to be done is to prove Lemma~\ref{le:inductive}.
To do this, we need to develop some further terminology and tools.
Let $x_i$, for $1 \leq i \leq k$,  be the horizontal distance between point
$a$ and the E side of $S_i$, respectively. We also set $x_0 = 0$.
\begin{definition}
A square $S_i$ has {\em potential} if
$$d_T(a,h_i) + d_T(a, l_i) + d_{S_i}(h_i, l_i) ~\le~ 4x_i$$
where $d_{S_i}(h_i,l_i)$ is the Euclidean distance when moving from
$h_i$ to $l_i$ along the sides of $S_i$, clockwise.
\end{definition}

\begin{lemma}
\label{le:potential}
If $R(a,b)$ is empty then $S_1$ has potential. Furthemore, for any
$1 \leq i < k$, if $S_i$ has potential but is not inductive then
$S_{i+1}$ has potential.
\end{lemma}

\begin{proof}
If $R(a,b)$ is empty then, by Lemma~\ref{le:properties}, $a$ lies on the W
side of $S_1$ and $x_1$ is the side length of square $S_1$. Also, $h_1$ lies
on the N or E side of $S_1$, and $l_1$ lies on
the S or E side of $S_1$. Then $d_T(a,h_1) + d_T(a, l_1) + d_{S_1}(h_1, l_1)$
is bounded by the perimeter of square $S_1$ which is $4x_1$.

Now assume that $S_i$, for $1 \leq i < k$, has potential but is not inductive.
Squares $S_i$ and $S_{i+1}$ both contain points $l_i$ and $h_i$. Because
$S_i$ is not inductive, edge $(l_i, h_i)$ must be steep and
thus $d_x(l_i, h_i) < d_y(l_i, h_i)$. To simplify the arguments, we assume
that $l_i$ is to the W of $h_i$, i.e., $x_{l_i} < x_{h_i}$. The case
$x_{l_i} > x_{h_i}$ can be shown using equivalent arguments.  

Since $T_i = \triangle(h_{i-1},h_i,l_{i-1}=l_i)$ or 
$T_i = \triangle(h_{i-1}=h_i,l_{i-1},l_i)$, there has to be a side of $S_i$
between the sides on which $l_i$ and $h_i$ lie, when moving clockwise from
$l_i$ to $h_i$. Using the constraints on the position of $h_i$ and $l_i$
within $S_i$ from Lemma~\ref{le:properties} and using assumptions that
$(l_i,h_i)$ is steep and that $x_{l_i} < x_{h_i}$, we deduce that $l_i$ must be
on the S side and $h_i$ must be on the N or E side of $S_i$.

If $h_i$ is on the N side of $S_i$ then, because $x_{l_i} < x_{h_i}$, $h_i$
must also be on the N side of $S_{i+1}$ and either $l_i$ is on the S side of
$S_{i+1}$ and
\begin{equation}
\label{eq:translation}
d_{S_{i+1}}(h_i, l_i) - d_{S_i}(h_i , l_i) ~=~ 2(x_{i+1}-x_i)
\end{equation}
as shown in Figure~\ref{fi:cases}a) or $l_i$ is on the W side of
$S_{i+1}$, in which case
\begin{equation}
\label{eq:growth}
d_{S_{i+1}}(h_i, l_i) - d_{S_i}(h_i , l_i) ~\le~ 4(x_{i+1}-x_i)
\end{equation}
as shown in Figure~\ref{fi:cases}b).

If $h_i$ is on the E side of $S_i$ then, because $x_{i+1} > x_i$ (since
$(l_i,h_i)$ is steep), $h_i$ must be on the N side of $S_{i+1}$ and either
$l_i$ is on the S side of $S_{i+1}$ and inequality~(\ref{eq:translation}) holds
or $l_i$ is on the W side of $S_{i+1}$ and inequality~(\ref{eq:growth})
holds, as shown in Figure~\ref{fi:cases}c).

Since $S_i$ has potential, in all cases we obtain:
\begin{eqnarray}
\label{eq:invariant}
d_T(a,h_i) + d_T(a, l_i) + d_{S_{i+1}}(h_i, l_i) ~\le~ 4x_{i+1}~.
\end{eqnarray}

Assume $T_{i+1} = \triangle(h_i, h_{i+1}, l_i=l_{i+1})$; in other words,
$(h_i,h_{i+1})$ is an edge of $T$ with $h_{i+1}$ lying somewhere on the boundary
of $S_{i+1}$ beween $h_i$ and $l_i$, when moving clockwise from $h_i$ to $l_i$.
By the triangular inequality, $d_2(h_i, h_{i+1}) \leq d_{S_{i+1}}(h_i, h_{i+1})$
and we have that:
\begin{eqnarray*}
d_T(a,h_{i+1}) + d_T(a, l_{i+1}) + d_{S_{i+1}}(h_{i+1},l_{i+1}) & \leq & d_T(a,h_i) + d_T(a, l_i) + d_2(h_i, h_{i+1}) + d_{S_{i+1}}(h_{i+1},l_i) \\
 & \leq & d_T(a,h_i) + d_T(a, l_i) + d_{S_{i+1}}(h_i, l_i) \\
 & \leq & 4x_{i+1}~.
\end{eqnarray*}
Thus $S_{i+1}$ has potential. The argument for the case when
$T_{i+1} = \triangle(h_i=h_{i+1}, l_i, l_{i+1})$ is symmetric.
\end{proof}

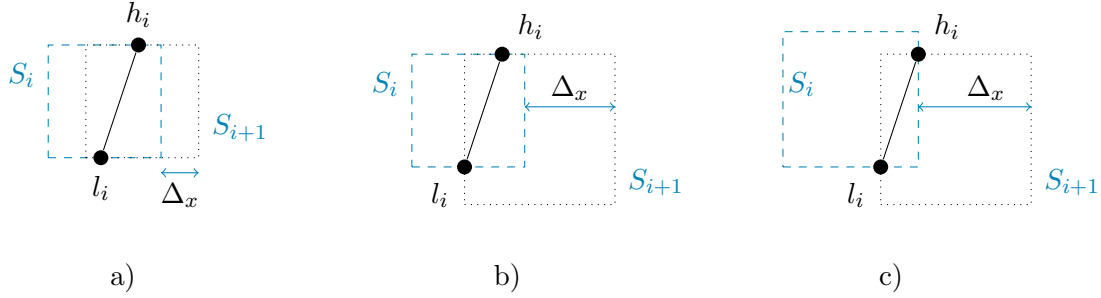
\begin{figure}
\begin{center}
\begin{tikzpicture}

\draw [color=textcolor, dashed] (0,0) rectangle (1.5,1.5);
\draw [dotted] (0.5,0) rectangle (2,1.5);

\draw (0,1.1) node [inner sep=1pt, label=180:{\color{textcolor}{$S_i$}}] {};
\draw (2,0.4) node [inner sep=1pt, label=0:{\color{textcolor}{$S_{i+1}$}}] {};

\node at (1.2,1.5) [fill, circle, inner sep=2pt, label=90:$h_i$] (hi) {};
\node at (0.7,0) [fill, circle, inner sep=2pt, label=-90:$l_i$] (li) {};

\draw [color=textcolor, <->] (1.5,-0.2) -- (2,-0.2);
\node at (1.75, -0.5) {$\Delta_x$};

\draw (hi) -- (li);
\end{tikzpicture}\hspace{1cm}
\begin{tikzpicture}

\draw [color=textcolor, dashed] (0,0) rectangle (1.5,1.5);
\draw [dotted] (0.7,-0.5) rectangle (2.7,1.5);

\draw (0,1.1) node [inner sep=1pt, label=180:{\color{textcolor}{$S_i$}}] {};
\draw (2.7,-0.2) node [inner sep=1pt, label=0:{\color{textcolor}{$S_{i+1}$}}] {};

\node at (1.2,1.5) [fill, circle, inner sep=2pt, label=45:$h_i$] (hi) {};
\node at (0.7,0) [fill, circle, inner sep=2pt, label=-135:$l_i$] (li) {};

\draw [color=textcolor, <->] (1.5,0.8) -- (2.7,0.8);
\node at (2.1, 1.05) {$\Delta_x$};


\draw (hi) -- (li);
\end{tikzpicture}\hspace{1cm}
\begin{tikzpicture}

\draw [color=textcolor, dashed] (-0.6,0) rectangle (1.2,1.8);
\draw [dotted] (0.7,-0.5) rectangle (2.7,1.5);

\draw (0,1.1) node [inner sep=1pt, label=180:{\color{textcolor}{$S_i$}}] {};
\draw (2.7,-0.2) node [inner sep=1pt, label=0:{\color{textcolor}{$S_{i+1}$}}] {};

\node at (1.2,1.5) [fill, circle, inner sep=2pt, label=45:$h_i$] (hi) {};
\node at (0.7,0) [fill, circle, inner sep=2pt, label=-135:$l_i$] (li) {};

\draw [color=textcolor, <->] (1.2,0.8) -- (2.7,0.8);
\node at (2.1, 1.05) {$\Delta_x$};


\draw (hi) -- (li);
\end{tikzpicture}
\end{center}
\hspace{1.75cm} a) \hspace{4.5cm} b)  \hspace{4.5cm} c)
\caption{The first, second and fourth case in the proof of
Lemma~\ref{le:potential}. In each case, the difference $d_{S_{i+1}}(h_i, l_i)
- d_{S_i}(h_i,l_i)$ is shown to be at most $4\Delta_x$, where $\Delta_x =
x_{i+1}-x_i$.}
\label{fi:cases}
\end{figure}

\begin{definition}
  A vertex $c$ ($h_i$ or $l_i$) of $T_i$ is \emph{promising in $S_i$}
  if it lies on the E side of $S_i$.
\end{definition}

\begin{lemma}
  \label{le:path}
  If square $S_i$ has potential and $c = h_i$ or $c = l_i$ is a
  promising point in $S_i$ then $$d_T(a, c) ~\le~ 2x_c~.$$
\end{lemma}

\begin{proof}
  W.l.o.g., assume $c = h_i$. Since $h_i$ is promising, $x_c = x_{h_i}
  = x_i$.  Because $S_i$ has potential, either $d_T(a,h_i) \leq
  2x_{h_i}$ or $d_T(a,l_i) + d_{S_i}(l_i,h_i) \leq 2x_{h_i}$. In the
  second case, we can use edge $(l_i,h_i)$ and the triangular
  inequality to obtain $d_T(a,h_i) \leq d_T(a,l_i) + |l_ih_i| \leq
  2x_{h_i}$.
\end{proof}

Here we define the maximal high and minimal low path.

\begin{definition}~
  \begin{itemize}
  \item If $h_j$ is promising in $S_j$, the \emph{maximal high path
      ending at} $h_j$ is simply $h_j$; otherwise, it is the path
    $h_i, h_{i+1}, \dots, h_j$ such that $h_{i+1},\dots,h_j$ are not
    promising and either $i=0$ or $h_i$ is promising in $S_i$.
  \item If $l_j$ is promising in $S_j$, the \emph{maximal low path
      ending at} $l_j$ is simply $l_j$; otherwise, it is the path
    $l_i, l_{i+1}, \dots, l_j$ such that $l_{i+1},\dots,l_j$ are not
    promising and either $i=0$ or $l_i$ is promising in $S_i$.
  \end{itemize}
\end{definition}

Note that by Lemma~\ref{le:properties}, all edges on the path $h_i,
h_{i+1}, \dots, h_j$ are WN edges and thus the path length is bounded
by $(x_{h_j}-x_{h_i}) + (y_{h_j}-y_{h_i})$. Similarly, all edges in
path $l_i, l_{i+1}, \dots, l_{j}$ are WS edges and the length of the
path is at most $(x_{l_j}-x_{l_i}) + (y_{l_i}-y_{l_j})$.

We now have the tools to prove Lemma~\ref{le:inductive}.

\begin{proofof}{\bf Lemma~\ref{le:inductive}.}
If $R(a,b)$ is empty then, by Lemma~\ref{le:properties}, $b$ is promising.
Thus, by Lemma~\ref{le:potential} and Lemma~\ref{le:path}, if no square
$S_1, \dots, S_k$ is inductive then $d_T(a,b) \leq 2x < (1+\sqrt{2})x + y$.

Assume now that there is at least one inductive square in the sequence of
squares $S_1, \dots, S_k$. Let $S_j$ be the first inductive square and
assume, for now, that $h_j$ is the inductive point in $S_j$. By
Lemma~\ref{le:potential}, every square $S_i$, for $i<j$, is a
potential square.

Since $(l_j,h_j)$ is gentle, it follows that $d_2(l_j,h_j) \leq
\sqrt{2}(x_{h_j} -
x_{l_j})$. Let $l_i, l_{i+1}, \dots, l_{j-1}=l_j$ be the maximal low
path ending at $l_j$. Note that $d_T(l_i, l_j) \leq (x_{l_j}-x_{l_i}) +
(y_{l_i}-y_{l_j})$. Either $l_i = l_0 = a$ or $l_i$ is a promising point in
potential square $S_i$; either way, by Lemma~\ref{le:path}, we have that
$d_T(a, l_i) \leq 2x_{l_i}$. Putting all this together, we get
\begin{eqnarray*}
d_T(a, h_j) + (y_{h_j}-y) & \leq & d_T(a, l_i) + d_T(l_i,l_j) + d_2(l_j, h_j) + y_{h_j} \\
                & \leq & 2x_{l_i} + (x_{l_j} - x_{l_i}) + (y_{l_i} - y_{l_j}) +
\sqrt{2}(x_{h_j} - x_{l_j}) + y_{h_j} \\
                & \leq & \sqrt{2} x_{h_j} + x_{l_j} + y_{h_j} - y_{l_j} \\
                & \leq & (1 + \sqrt{2}) x_{h_j}
\end{eqnarray*}
where the last inequality follows $x_{l_j} + y_{h_j} - y_{l_j}\leq
x_{h_j}$, i.e., from the assumption that edge $(l_j,h_j)$ is gentle.

If, instead, $c = l_j$ is the inductive point in inductive square $S_j$,
let $h_i, h_{i+1}, \dots, h_{j-1} = h_j$ be the maximal high path
ending at $h_j$. Then $d_T(h_i, h_j) \leq (x_{h_j}-x_{h_i}) + (y_{h_j}-y_{h_i})$.
Just as in the first case, we have that $d_T(a, h_i) \leq 2x_{h_i}$ and
\begin{eqnarray*}
d_T(a, l_j) - y_{l_j} & \leq & d_T(a, h_i) + d_T(h_i,h_j) + d_2(h_j, l_j) - y_{l_j} \\
                  & \leq & 2 x_{h_i} + (x_{h_j}-x_{h_i}) + (y_{h_j}-y_{h_i}) + \sqrt{2}(x_{l_j}-x_{h_j}) - y_{l_j} \\
                  & \leq & \sqrt{2} x_{l_j} + x_{h_j} + y_{h_j} - y_{l_j} \\
                  & \leq & (1 + \sqrt{2}) x_{l_j}~. 
\end{eqnarray*}
where the last inequality follows from $x_{h_j} + y_{h_j} - y_{l_j} \leq
x_{l_j}$, i.e., from the assumption that $(h_j, l_j)$ is gentle.
\end{proofof}

\section{Conclusion and perspectives}





The $L_1$-Delaunay triangulation is the first type of Delaunay
triangulation to be shown to be a spanner~\cite{Chew86}. Progress on
the spanning properties of the TD-Delaunay and the classical
$L_2$-Delaunay triangulation soon followed.  In this paper, we
determine the precise stretch factor of an $L_1$- and
$L_\infty$-Delaunay triangulation and close the problem for good. The
techniques we develop in this paper have potential to be successfully
applied to Delaunay triangulations defined by other regular polygons,
and possibly even to the classical Delaunay triangulation.


From a routing perspective, it is of interest to construct routes in
geometric graphs that can be determined \emph{locally} from a
neighbor's coordinates only~\cite{BCD09}. Unfortunately, the route
that is implicitely constructed in our proof is built using non-local
decisions. It would be interesting to know whether in the
$L_1$-Delaunay triangulation a route with stretch $\sqrt{4+2\sqrt{2}}$
can be constructed using a local routing algorithm. For TD-Delaunay
triangulations, \cite{BFvRV12} showed that there is no local routing
algorithm that achieves a stretch that is less than $5/\sqrt3 \approx
2.88$, whereas the stretch factor is actually~$2$.  We leave open the
questions regarding the gap between the stretch factor of
$L_1$-Delaunay triangulations and the stretch that is possible using
local routing.



\bibliographystyle{alpha}
\bibliography{bib}

\end{document}